\documentclass[11pt,a4paper]{amsart}
\usepackage{amsmath,amssymb,amsthm,amscd,verbatim,enumerate}
\usepackage{graphicx,subfigure}
\usepackage[colorlinks=true,citecolor=black,linkcolor=black,urlcolor=blue]{hyperref}
\usepackage[lmargin=31mm,rmargin=31mm,bmargin=31mm,tmargin=31mm]{geometry}

\renewcommand{\baselinestretch}{1.1}
\allowdisplaybreaks

\theoremstyle{plain}
\newtheorem{theorem}{Theorem}[section]
\newtheorem{lemma}[theorem]{Lemma}
\newtheorem{corollary}[theorem]{Corollary}

\newtheorem{claim}[theorem]{Claim}
\theoremstyle{definition}

\renewcommand{\leq}{\leqslant}

\renewcommand{\geq}{\geqslant}

\newcommand{\DEF}{\sl}    
\newcommand{\sig}{\Sigma_{2}^{\text{P}}}

\begin{document}

\title{Hitting All Maximal Independent Sets of a Bipartite Graph}

\author{Jean Cardinal} 
\author{Gwena\"el Joret}
\address{\newline  D\'epartement d'Informatique 
\newline Universit\'e Libre de Bruxelles
\newline Brussels, Belgium}
\email{\{jcardin, gjoret\}@ulb.ac.be}

\thanks{Gwena\"el Joret is a Postdoctoral Researcher of the Fonds
 National de la Recherche Scientifique (F.R.S.--FNRS)}

\date{\today}
\sloppy
\maketitle

\begin{abstract}
We prove that given a bipartite graph $G$ with vertex set $V$ and an integer $k$, deciding whether there exists a subset of $V$ of size at most $k$ hitting all maximal independent sets of $G$ is complete for the class $\sig$. 
\end{abstract}

\section{Introduction}

A maximal independent set in a graph is a subset of pairwise nonadjacent vertices that is maximal with 
respect to inclusion. We consider the following problem: given a graph $G$ with vertex set $V$ and an integer $k$, 
does there exist a subset of $V$ of size at most $k$ that intersects every maximal independent set of $G$?
While this problem is well known to be NP-hard (see for instance~\cite{DuffusKT91}), it is not known to belong to NP, hence not known to be NP-complete. 
We prove that this is unlikely to be the case: the problem is complete for the class $\sig = \text{NP}^{\text{NP}}$, at the second level 
of the polynomial hierarchy. Furthermore, our proof holds even in the case where the input graph is bipartite.

\subsection{Previous works}
We distinguish three lines of work that are closely related to our result. They concern respectively the clique transversal problem, 
fibres in partially ordered sets, and the clique coloring problem.

\subsubsection{Clique transversals.}

A {\DEF clique} is a subset of pairwise adjacent vertices.
The problem we consider is often formulated in the complement graph, where we wish to hit 
every maximal clique with a subset of the vertices. 
Such a subset is called a {\DEF clique transversal}, and the corresponding clique transversal problem consists of 
deciding whether there exists a clique transversal of size at most $k$.

Early investigations can be found in Aigner and Andreae~\cite{AA86}, Tuza~\cite{Tuza90}, and 
Erd\"{o}s, Gallai, and Tuza~\cite{ErdosGT92}. The latter contains the first NP-hardness proof for the problem. 
The authors also showed that there are graphs on $n$ vertices where the minimum size of a clique transversal 
is as large as $n - o(n)$. 

Balachandran, Nagavamsi, and Rangan~\cite{BalachandranNR96} observed that the minimum clique transversal problem was solvable 
in polynomial time on comparability graphs, a consequence of Menger's theorem.
Clique transversals of line graphs and chordal graphs have been considered by 
Andreae, Schughart, and Tuza~\cite{AndreaeST91}, 
and Andreae~\cite{Andreae98}, respectively.

Later, Guruswami and Pandu Rangan~\cite{GuruswamiR00} proved the NP-hardness of the clique transversal problem in a number of 
restricted cases, including 
planar graphs and line graphs. 
They also proved it remains NP-hard on complements of bipartite graphs, which corresponds exactly to our problem, 
and provided polynomial-time algorithms for some other cases, including strongly chordal graphs and 
Helly circular-arc graphs.

Dur\'{a}n, Lin, and Szwarcfiter~\cite{DuranLS02}  
studied clique-perfect graphs, defined as graphs for which
the minimum size of a clique transversal is equal to the maximum number of pairwise disjoint maximal cliques.

More recently, Dur{\'a}n, Lin, Mera, and Szwarcfiter~\cite{DuranLMS08}  
gave a more efficient algorithm for finding a minimum
clique transversal of a Helly circular arc graph. 
The reader is referred to the paper~\cite{DuranLMS08} for more references. 

\subsubsection{Fibres in posets.}
The notion of fibre of a poset is closely related to the clique transversal problem. 
Given a partially ordered set $P=(X,\leq)$, a fibre is a subset of $X$ that intersects every maximal antichain of $P$, that is, 
every maximal subset of pairwise incomparable elements. Hence fibres are clique transversals of the complement of the comparability
graph of the poset.

Aigner and Andreae~\cite{AA86} asked the question of whether there always existed a fibre of size at most half the number of elements, provided
the poset has no splitting element. Lonc and Rival~\cite{LoncR87} further conjectured that the elements of any finite poset without
splitting element could be partitioned into two fibres. This was answered in the negative by 
Duffus, Sands, Sauer, and Woodrow~\cite{DuffusSSW91}. 
Duffus, Kierstead, and Trotter~\cite{DuffusKT91} 
proved, however, that every poset without splitting element has a fibre of size at most 2/3 its number of elements. 
They also provided a coNP-completeness proof for the verification problem, consisting of deciding whether a given subset is a fibre. This result holds
even if the poset has height two, which again corresponds exactly to our setting, since every bipartite graph is the comparability graph of a poset
of height two. Finally, they proved that the problem of deciding whether there exists a fibre of a given size is NP-hard.

Our proof borrows ideas from the complexity results of Duffus et al.~\cite{DuffusKT91}. Our result can actually be reformulated as follows:
given a poset $P$ and an integer $k$, deciding whether $P$ has a fibre of size at most $k$ is $\sig$-complete, even if $P$ has height two.

\subsubsection{Clique coloring.}
Marx~\cite{Marx11} recently published a proof that the following problem is complete for $\sig$: 
Given a graph, is it possible to color its vertices with two
colors so that no maximal clique is monochromatic? Hence we wish to decide whether there exists a partition of the
vertices into two clique transversals. Our hardness proof uses the same general structure as Marx's reduction.  

Other contributions on the clique coloring problem include for instance Kratochv\'{i}l and Tusza~\cite{KT02} 
and Bacs{\'o}, Gravier, Gy\'arf\'as, Preissmann, and Seb\"o~\cite{BacsoGGPS04}.  

\subsection{Plan}

In the next section, we briefly review different definitions of the class $\sig$ and give examples of problems that are known to be $\sig$-complete.
In Section~\ref{sec:main}, we give the proof of our main result. 

\section{The Polynomial Hierarchy and the class $\sig$}

We briefly review standard material on the complexity classes defining the polynomial hierarchy. The reader is referred to the textbook of Papadimitriou~\cite{Pbook}
for a more extensive treatment, and to the compilation of Schaefer and Umans~\cite{SU02} for more examples of problems that are complete for
some classes of the polynomial hierarchy.

An {\DEF oracle} for a language $L$ can be thought of as a subroutine that checks whether a given word belongs to $L$ in constant time.
When a Turing machine is given access to an oracle for some language that does not belong to P, the class of problems it can solve in
polynomial time is possibly increased. The so-called polynomial hierarchy (PH) is a hierarchy of complexity classes that involves
oracles for languages in NP. Note that in that case it does not matter which particular language $L$ is involved as long as it
is NP-complete. Also note that oracles for NP and coNP yield the same computational power.

The class $\sig$ we will be interested in is also denoted $\text{NP}^{\text{NP}}$, where the exponent means that the machine involved
has access to an NP oracle. Hence this is the class of languages, or problems, that can be decided in polynomial time on a nondeterministic
Turing machine with access to an NP oracle. Similarly, the class $\Pi_2^P$ is defined as $\text{coNP}^{\text{NP}}$, and contains exactly the 
complements of the languages in $\sig$.

In general, PH is the union of $\Delta_0^{\text{P}} = \Sigma_0^{\text{P}} = \Pi_0^{\text{P}} = \text{P}$ and the following classes defined recursively for $i\geq 1$ :
$\Delta_i^{\text{P}} = \text{P}^{\Sigma_{i-1}^{\text{P}}}$, $\Sigma_i^{\text{P}} = \text{NP}^{\Sigma_{i-1}^{\text{P}}}$, and $\Pi_i^{\text{P}} = \text{coNP}^{\Sigma_{i-1}^{\text{P}}}$.
So for instance $\Sigma_1^{\text{P}}=$NP, $\Pi_1^{\text{P}}=$coNP, and $\Pi_2^{\text{P}}=\text{coNP}^{\text{NP}}$.

An alternative definition of the classes $\Sigma_i^{\text{P}}$ and $\Pi_i^{\text{P}}$ makes use of canonical problems generalizing
the NP-complete satisfiability (SAT) problem. Recall that the SAT problem consists of deciding whether $\exists x f(x)$ holds, where $f(x)$ is a boolean formula
defined on the variables $x=(x_1, x_2, \ldots ,x_n)$. It can be assumed without loss of generality that $f$ is given in 3-conjunctive normal form (3-CNF), 
yielding the well-known NP-complete 3SAT problem. It can be shown that the problem of 
deciding whether $\exists x^1 \forall x^2 \exists x^3 \ldots Q_i x^i f(x^1,x^2,\ldots ,x^i)$, where $f$ is a boolean formula, the $x^j$ are disjoint sets of boolean variables, and the existential and universal quantifiers alternate, is complete for $\Sigma_i^{\text{P}}$ (the quantifier $Q_i$ for $x^i$ is universal if $i$ is even, and existential otherwise). This was originally proved by Stockmeyer~\cite{S76} and Wrathall~\cite{W76}. 
Similarly, the problem $\forall x^1 \exists x^2 \forall x^3 \ldots Q_i x^i f(x^1,x^2,\ldots ,x^i)$ 
is complete for $\Pi_i^{\text{P}}$ (here $Q_i$ is universal if $i$ is odd, and existential otherwise). 

We will use the corresponding problem for $\sig$ where $f$ is a disjunction of terms
of size three, or in other words in 3-disjunctive normal form (3-DNF).
A {\DEF Q-3-DNF formula} is a formula $\varphi=\varphi(x, y)$  
with $x = (x_{1}, \dots, x_{n})$ and $y = (y_{1}, \dots, y_{m})$ where
the $x_{i}$'s and $y_{j}$'s are distinct variables, which is of the form 
$$
\exists x \; \forall y \; (t_{1,1} \land t_{1,2} \land t_{1,3}) \lor
(t_{2,1} \land t_{2,2} \land t_{2,3}) \lor \cdots \lor
(t_{q,1} \land t_{q,2} \land t_{q,3})
$$
where, for each $\ell \in \{1,2,\ldots ,q\}$, the term 
$(t_{\ell,1} \land t_{\ell,2} \land t_{\ell,3})$ consists of 
three literals corresponding to variables in $\{x_{1}, \dots, x_{n}\}\cup \{y_{1}, \dots, y_{m}\}$. 
Deciding whether $\varphi$ holds true consists of checking whether there exists a
boolean vector $x$ such that for any boolean vector $y$, there is at least one term
$(t_{\ell, 1}\land t_{\ell, 2}\land t_{\ell, 3})$ that evaluates to true.

\begin{lemma}[Stockmeyer~\cite{S76}; Wrathall~\cite{W76}]
\label{lem:q3dnf}
It is $\sig$-complete to decide whether a Q-3-DNF formula $\varphi$ holds true.
\end{lemma}

\section{Proof of the Main Result}
\label{sec:main}

We first show that a restricted version of the canonical problem for $\sig$ remains complete for $\sig$.
We then reduce from this problem in our hardness proof. 
We use the notation $[n] := \{1,2,\ldots ,n\}$, where $n$ is any natural number.

\subsection{Nice and Monotone Formulas}

Let $\varphi (x,y)$ be a Q-3-DNF formula, with $x = (x_{1}, \dots, x_{n})$, $y = (y_{1}, \dots, y_{m})$,
and comprising $q$ terms of the form $(t_{\ell,1} \land t_{\ell,2} \land t_{\ell,3})$, $\ell\in [q]$.
The  formula $\varphi$ is said to be {\DEF monotone} if 
the three literals $t_{\ell,1}, t_{\ell,2}, t_{\ell,3}$  are either all positive 
or all negative for each $\ell \in [q]$; in this case 
a term of $\varphi$ is said to be {\DEF positive} ({\DEF negative}) if its three literals
are positive (negative, respectively). Moreover, if  
$\varphi$ is such that, for each $i\in [n]$,   
there exist a positive term and a negative term avoiding the literals $x_{i}$ 
and $\bar x_{i}$, respectively, then $\varphi$ is said to be {\DEF nice}. 
(Note that such formulas have at least one positive term and one negative term.)

\begin{lemma}
\label{lem:monotone}
Given a nice monotone Q-3-DNF formula $\varphi$, it is $\sig$-complete to decide whether $\varphi$ holds true or not. 
\end{lemma}
\begin{proof}
The problem clearly belongs to $\sig$, since it is a special case of the 
problem of deciding whether a Q-3-DNF formula holds true, which from Lemma~\ref{lem:q3dnf} is complete for $\sig$.  
For the hardness part, we reduce from the latter problem. 
Let thus $\psi = \psi(x,y)$ be a given Q-3-DNF formula. 

We first transform $\psi$ into a formula containing only monotone terms.
For this, we can apply the same transformation as the one proposed by Gold~\cite{Gold78}
for proving NP-hardness of the monotone SAT problem. For each non-monotone term in
$\psi$, we add an extra variable in the group of variables that are quantified universally,
and split the term into two new terms. Consider such a term and suppose without loss of generality
that it has the form $(a\wedge b\wedge \bar{c})$. (The case where there are two negative literals is
symmetric.) Then we can split it into two terms as follows
$$
(a\wedge b\wedge \bar{c}) \leftrightarrow \forall d\ (a\wedge b\wedge d) \vee (\bar{d}\wedge \bar{c}) .
$$
Doing this for every non-monotone term yields a monotone formula that holds true if and only if $\psi$ holds true.

Now assume $\psi = \psi(x,y)$ is a formula of the above form; that is, the formula is monotone and 
each term has size two or three. 
We now make the formula nice. For each literal $x_i$ such that there is no negative term avoiding $\bar{x}_i$, include a new dummy term
of the form $(x_i\wedge z\wedge \bar{z})$, where $z$ is a new variable that is quantified universally. 
This obviously does not alter the validity of the formula. This term can in turn be split into two monotone terms using
a second additional variable $w$:
$$
(x_i\wedge z\wedge \bar{z}) \leftrightarrow \forall w\ (x_i\wedge z\wedge w) \vee (\bar{z}\wedge \bar{w}).
$$   
Those two new terms are monotone, and the second term is negative and 
does not contain $\bar{x}_i$, as wanted. Apply the symmetric operation 
each literal literal $x_i$ such that there is no positive term avoiding ${x}_i$. 
This ensures that the resulting formula is nice. 

Finally, each term $(a \land b)$ composed 
of two literals can be transformed into a term $(a \land b \land c)$ of size three, where $c$
is a new variable quantified existentially, so that the resulting formula is simultaneously monotone, nice, and Q-3-DNF. 
\end{proof}

\subsection{Main Result}

For simplicity, a vertex-subset of a graph $G$ meeting all maximal 
independent sets of $G$ will be called a {\DEF transversal} of $G$. 

\begin{theorem}
\label{thm:main}
Given a bipartite graph $G$ and a positive integer $k$, it is 
$\sig$-complete to decide whether $G$ has a transversal of size at most $k$. 
\end{theorem}
\begin{proof}
First, we notice that the problem belongs to $\sig$. This poses no difficulty, as it is known that deciding whether 
a given subset is a transversal is coNP-complete, even if the graph is bipartite~\cite{DuffusKT91}. 
Hence the whole problem can be solved on a non-deterministic Turing machine
with access to an oracle for deciding an NP-complete language. This is exactly the definition of $\sig$.

For the hardness part, we reduce from the problem of deciding whether a given nice monotone Q-3-DNF formula 
holds true or not, which is $\sig$-complete by Lemma~\ref{lem:monotone}. 
Let $\varphi=\varphi(x, y)$ be such a formula, with 
$x = (x_{1}, \dots, x_{n})$, $y = (y_{1}, \dots, y_{m})$, 
positive terms $(t_{1,1} \land t_{1,2} \land t_{1,3}), \dots, (t_{q,1} \land t_{q,2} \land t_{q,3})$, 
and negative terms 
$(t'_{1,1} \land t'_{1,2} \land t'_{1,3}), \dots, (t'_{q',1} \land t'_{q',2} \land t'_{q',3})$.  

We construct a bipartite graph $G$ from $\varphi$ as follows. The vertex set of $G$ is 
composed of:
\begin{itemize}
\item vertices $a_{i}, b_{i}, x_{i}, \bar a_{i}, \bar b_{i}, \bar x_{i}$ for each 
variable $x_{i}$ ($i\in [n]$); 
\item vertices $y_{j}$ and $\bar y_{j}$ for each 
variable $y_{j}$ ($j\in [m]$); 
\item vertices $t_{\ell}, r_{\ell}, s_{\ell}$ for each positive term
$(t_{\ell,1} \land t_{\ell,2} \land t_{\ell,3})$ ($\ell \in [q]$), and
\item vertices $t'_{\ell}, r'_{\ell}, s'_{\ell}$ for each negative term
$(t'_{\ell,1} \land t'_{\ell,2} \land t'_{\ell,3})$ ($\ell \in [q']$). 
\end{itemize}
The bipartition of $G$ we will consider is $(P, N)$ with 
\begin{align*}
P =\, & \{a_{i}, b_{i}, x_{i}: i\in [n]\}
\cup \{y_{1}, \dots, y_{m}\} 
\cup \{t_{\ell}, r_{\ell}: \ell \in [q]\} \\
& \cup \{s'_{\ell}: \ell \in [q']\}
\end{align*}
and $N := V(G) - P$. ($P$ stands for `positive' and $N$ for `negative'.) 
The edges of $G$ are determined as follows:
\begin{itemize}
\item for each variable $x_{i}$, add the edge $x_{i}\bar x_{i}$; 
\item for each variable $y_{j}$, add the edge $y_{j}\bar y_{j}$; 
\item for each $i\in [n]$, link $a_{i}$ and $b_{i}$ to each vertex in 
$N - \{\bar x_{i}, \bar b_{i}\}$ and $N - \{\bar a_{i}, \bar b_{i}\}$, respectively, 
and $\bar a_{i}$ and $\bar b_{i}$ to each vertex in 
$P - \{x_{i}, b_{i}\}$ and $P - \{a_{i}, b_{i}\}$, respectively; 
\item for each positive term $(t_{\ell,1} \land t_{\ell,2} \land t_{\ell,3})$, 
add the $3$ edges $t_{\ell}\bar x_{i_{1}}, t_{\ell}\bar x_{i_{2}}, t_{\ell}\bar x_{i_{3}}$  
where $x_{i_{d}} = t_{\ell,d}$ for each $d\in \{1,2,3\}$; 
\item for each negative term $(t'_{\ell,1} \land t'_{\ell,2} \land t'_{\ell,3})$, 
add the $3$ edges $t'_{\ell}x_{i_{1}}, t'_{\ell}x_{i_{2}}, t'_{\ell} x_{i_{3}}$  
where $\bar x_{i_{d}} = t'_{\ell,d}$ for each $d\in \{1,2,3\}$; 
\item for each $\ell \in [q]$, link $r_{\ell}$ and $s_{\ell}$ 
to each vertex in $N - \{s_{\ell}\}$ and $P - \{t_{\ell}, r_{\ell}\}$, respectively, and  
\item for each $\ell \in [q']$, link $r'_{\ell}$ and $s'_{\ell}$ 
to each vertex in $P - \{s'_{\ell}\}$ and $N - \{t'_{\ell}, r'_{\ell}\}$, respectively. 
\end{itemize}

\begin{figure}
\begin{center}
\subfigure[\label{fig:vargadget}Gadget for the variable $x_i$. The four maximal independent sets shown remain maximal independent sets in $G$.]{\includegraphics[scale=.6]{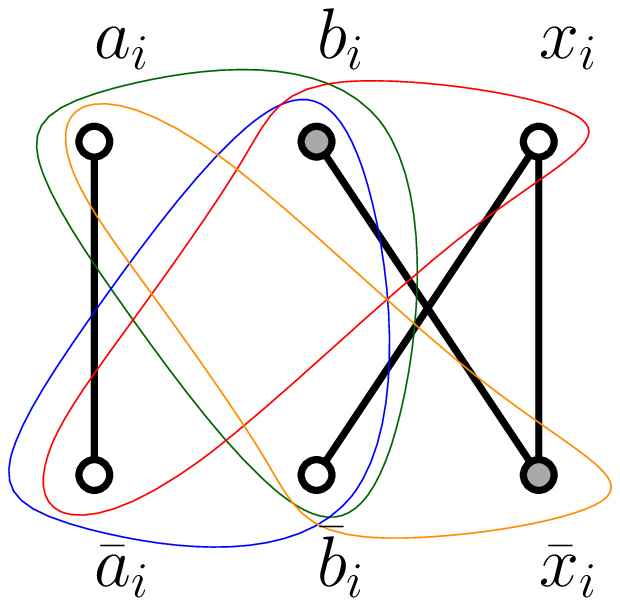}}
\hspace{2cm}
\subfigure[\label{fig:termgadget}Gadget for the $\ell$th positive term $(x_{i_1}\wedge x_{i_2}\wedge x_{i_3})$. The maximal independent set $\{r_{\ell},s_{\ell},t_{\ell}\}$ remains a maximal independent set in $G$, as the vertices $r_{\ell}$ and $s_{\ell}$ are adjacent to all other vertices in the opposite side of the bipartition.]{\includegraphics[scale=.6]{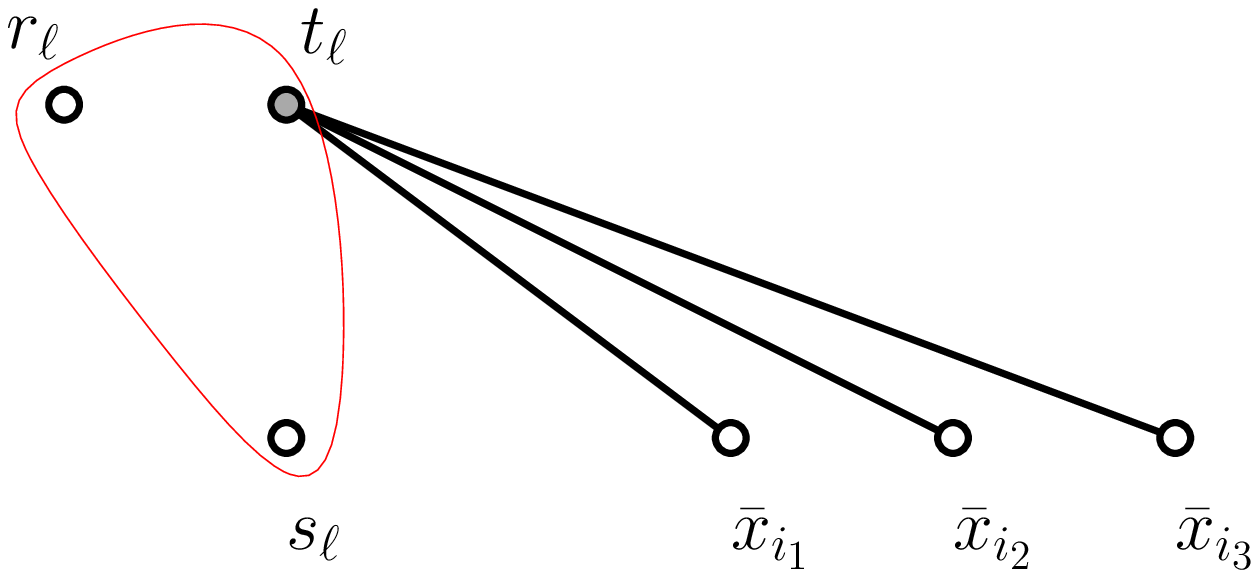}}
\end{center}
\caption{\label{fig:gadgets}Gadgets for the reduction.}
\end{figure}

Let us consider the maximal independent sets of $G$. 
Say that such a set is {\DEF regular} if it avoids the four vertices
$a_{i}, b_{i}, \bar a_{i}, \bar b_{i}$ for each $i \in [n]$,
the two vertices $r_{\ell}, s_{\ell}$ for each $\ell \in [q]$, and 
the two vertices $r'_{\ell}, s'_{\ell}$ for each $\ell \in [q']$.  
It will be helpful to have a characterization of 
the maximal independent sets of $G$ that are {\em not} regular.  

Let $S$ be such a set. First 
suppose that $S \cap \{a_{i}, b_{i}, \bar a_{i}, \bar b_{i}\} \neq \varnothing$ for
some $i \in [n]$. Recall that $a_i$ and $b_i$ are adjacent to all vertices in $N$, except, 
respectively to $\bar{x}_i, \bar{b}_i$, and $\bar{a}_i, \bar{b}_i$, and symmetrically for 
$\bar{a}_i$ and $\bar{b}_i$. Hence every maximal independent set containing, 
say, $a_i$ and $\bar{b}_i$ cannot be extended with vertices outside the variable gadget
defined by the vertices $a_i, b_i, x_i, \bar{a}_i,\bar{b}_i,\bar{x}_i$ (see figure~\ref{fig:vargadget}).
Therefore, upon inspection of $G$, one can check that either
\begin{itemize}
\item $S$ is equal to $P$ or $N$, or 
\item $S \subseteq \{a_{i}, b_{i}, \bar a_{i}, \bar b_{i}, x_{i}, \bar x_{i}\}$ in which case $S$ 
is one of the four sets $\{a_{i}, b_{i}, \bar b_{i}\}, \{\bar a_{i}, b_{i}, \bar b_{i}\}, 
\{a_{i}, \bar b_{i}, \bar x_{i}\}, \{\bar a_{i}, b_{i}, x_{i}\}$, or
\item  $S \cap \{a_{i}, b_{i}, \bar a_{i}, \bar b_{i}\} = \{a_{i}\}$ 
in which case $S$ is the union of  $\{a_{i}, \bar x_{i}\} \cup 
\{x_{1}, \dots x_{i-1}, x_{i+1}, \dots,  x_{n}\} \cup \{y_{1}, \dots, y_{m}\} =: S_{i}$  
with the set $T_{i}$ of all vertices $t_{\ell}$ ($\ell \in [q]$) 
such that the literal $x_{i}$ does not appear in the $\ell$-th positive term of $\varphi$, or 
\item  $S \cap \{a_{i}, b_{i}, \bar a_{i}, \bar b_{i}\} = \{\bar a_{i}\}$ 
in which case $S$ is the union of  $\{\bar a_{i}, x_{i}\} \cup 
\{\bar x_{1}, \dots \bar x_{i-1}, \bar x_{i+1}, \dots, \bar x_{n}\} 
\cup \{\bar y_{1}, \dots, \bar y_{m}\} =: S'_{i}$  
with the set $T'_{i}$ of all vertices $t'_{\ell}$ ($\ell \in [q']$) 
such that the literal $\bar x_{i}$ does not appear in the $\ell$-th negative term of $\varphi$.  
\end{itemize}
A crucial observation is that the sets $T_{i}$ and $T'_{i}$ defined above are non-empty 
for each $i\in [n]$; indeed, this exactly corresponds to the requirement 
that the formula $\varphi$ be nice. 

Now assume $S \cap \{a_{i}, b_{i}, \bar a_{i}, \bar b_{i}\} = \varnothing$ for
each $i \in [n]$. Since $S$ is not regular, we have $S \cap \{r_{\ell}, s_{\ell}\} \neq \varnothing$ for some $\ell \in [q]$, or 
$S \cap \{r'_{\ell}, s'_{\ell}\} \neq \varnothing$ for some $\ell \in [q']$. 
In the first case 
\begin{itemize}
\item either $S = \{t_{\ell}, r_{\ell}, s_{\ell}\}$, or
\item $S \cap \{t_{\ell}, r_{\ell}, s_{\ell}\} = \{t_{\ell}, s_{\ell}\}$, in which case
$S$ is the union of $\{t_{\ell}, s_{\ell}\} \cup \{t'_{1}, \dots, t'_{q'}\}=:W_{\ell}$ with 
the set $Z_{\ell}$ of all vertices in $\{\bar x_{1}, \dots, \bar x_{n}\} \cup \{\bar y_{1}, 
\dots, \bar y_{m}\}$ to which $t_{\ell}$ is not adjacent (that is, all negative literals 
that do not contradict the $\ell$-th positive term). 
\end{itemize}
In the second case the situation is symmetric as expected; that is, 
\begin{itemize}
\item either $S = \{t'_{\ell}, r'_{\ell}, s'_{\ell}\}$, or
\item $S$ is the union of $\{t'_{\ell}, s'_{\ell}\} \cup \{t_{1}, \dots, t_{q}\}=:W'_{\ell}$ with 
the set $Z_{\ell}$ of all vertices in $\{x_{1}, \dots, x_{n}\} \cup \{y_{1}, 
\dots, y_{m}\}$ to which $t'_{\ell}$ is not adjacent. 
\end{itemize}

Summarizing the above discussion, the maximal independent sets of $G$ can be classified as follows:
\begin{enumerate}
\item regular maximal independent sets of $G$;
\item $P$ and $N$;
\item \label{stab:vargadgets} $\{a_{i}, b_{i}, \bar b_{i}\}, \{\bar a_{i}, b_{i}, \bar b_{i}\}, 
\{a_{i}, \bar b_{i}, \bar x_{i}\}, \{\bar a_{i}, b_{i}, x_{i}\}$ for $i\in [n]$; 
\item \label{stab:pst} $S_{i}\cup T_{i}$ for $i\in [n]$; 
\item \label{stab:nst} $S'_{i}\cup T'_{i}$ for $i\in [n]$; 
\item \label{stab:ptermgadgets} $\{t_{\ell}, r_{\ell}, s_{\ell}\}$ for $\ell\in [q]$; 
\item \label{stab:pwz} $W_{\ell} \cup Z_{\ell}$ for $\ell\in [q]$; 
\item \label{stab:ntermgadgets} $\{t'_{\ell}, r'_{\ell}, s'_{\ell}\}$ for $\ell\in [q']$, and 
\item \label{stab:nwz} $W'_{\ell} \cup Z'_{\ell}$ for $\ell\in [q']$.  
\end{enumerate}

Let $k:= 2n + q + q'$. We are now ready to show that the formula $\varphi$ is true if and only if 
$G$ has a transversal of size at most $k$. 

{\bf ($\Rightarrow$)}
First suppose that $\varphi$ holds true.  
We have to show that there exists a transversal $X\subseteq V(G)$ of $G$ of size at most $k$.
We refer to the satisfying truth assignment for $x$ as $\mu: \{x_1, x_2, \ldots , x_n\}\mapsto \{\text{True}, \text{False}\}$.
We define $X$ as follows:
$$
X:= \{t_{\ell} : \ell\in [q]\}\cup \{t'_{\ell} : \ell\in [q']\} \cup \{b_i, \bar{x}_i :  \mu (x_i)=\text{True}\} \cup \{\bar{b}_i, x_i :  \mu (x_i)=\text{False}\} .
$$

The set $X$ has size exactly $k$. We now have to check that $X$ is actually a transversal. Since all vertices $t_{\ell}$ and $t'_{\ell}$ belong to $X$, the sets
$P$ and $N$, and every set of type \ref{stab:ptermgadgets} and \ref{stab:ntermgadgets} above are hit. For the same reason, 
the independent sets of type \ref{stab:pwz} and \ref{stab:nwz} are hit as well. Since the sets of the form 
$T_i$ and $T'_i$ are nonempty (see remark above), the independent sets of type \ref{stab:pst} and \ref{stab:nst} are also hit by the
vertices $t_{\ell}$ and $t'_{\ell}$ of $X$. And since $X$ contains either $\{b_i, \bar{x}_i\}$ or $\{\bar{b}_i, x_i\}$ for every $i\in [n]$, the independent sets
of type \ref{stab:vargadgets} are all hit. 

Thus it only remains to check whether the regular maximal independent sets are hit by $X$. Suppose this is not the case
and let $U\subseteq V(G)$ be a regular maximal independent set that does not intersect $X$. In particular, $U$ does not contain
any vertex of the form $t_{\ell}$ or $t'_{\ell}$. The set $U\cap \{x_1,\bar{x}_1,x_2,\bar{x}_2,\ldots ,\bar{x}_n,x_n\}$, by definition, 
contains exactly the literals that are set to True by the assignment $\mu$. Also, the intersection $U\cap \{y_1,\bar{y}_1,y_2,\bar{y}_2,\ldots ,\bar{y}_m,y_m\}$ 
can be interpreted as a truth assignment for $y$, as exactly one each pair of 
$y_j, \bar{y}_j$ ($j\in [m]$) must be contained in $U$. 
In this truth assignment for $y$, we set $y_j$ (respectively, $\bar{y}_j$) to True whenever it is contained in $U$. 
But now since this combined truth assignment is a truth assignment satisfying $\varphi$, there must exist a satisfied term. Let $t$ denote the vertex corresponding to this term; that is, $t= t_{\ell}$ 
if it is the $\ell$th positive term,  $t= t'_{\ell}$ if it is the $\ell$th negative one. 
The vertex $t$ is adjacent only to the three literals that contradict it. Since the term is satisfied, $t$ is not adjacent to any vertex of $U$. But then $U \cup \{t\}$ is again an independent set, contradicting the maximality of $U$. Hence such a $U$ cannot exist, and $X$ meets every maximal independent set of $G$. 

{\bf ($\Leftarrow$)}
Now assume that $G$ has a transversal $X$ of size at most $k$. In order to deduce a satisfying truth assignment for $x$, we first need to transform $X$ into some canonical form.

\begin{claim}
If $G$ has a transversal of size at most $k$, then it also has a transversal $X'$ of size exactly $k$, with the following properties:
\begin{itemize}
\item $\{t_{1}, \dots, t_{q}\}\subseteq X'$;
\item $\{t'_{1}, \dots, t'_{q'}\}\subseteq X'$, and
\item for each $i\in [n]$, either $\{\bar b_i, x_i\} \subseteq X'$ or $\{b_i, \bar x_i\} \subseteq X'$.
\end{itemize}
\end{claim} 
\begin{proof}
First we observe that the $k$ maximal independent sets 
$\{a_{i}, \bar b_{i}, \bar x_{i}\}$ ($i \in [n]$), 
$\{\bar a_{i}, b_{i}, x_{i}\}$ ($i \in [n]$), 
$\{t_{\ell}, r_{\ell}, s_{\ell}\}$ ($\ell\in [q]$), and 
$\{t'_{\ell}, r'_{\ell}, s'_{\ell}\}$ ($\ell\in [q']$) 
are all pairwise disjoint. Thus $|X| = k$, and 
$X$ meets each of them in exactly one vertex. Moreover, $|X \cap \{x_{i}, \bar x_{i}\}| \leq 1$ for each 
$i \in [n]$, since  $X$ would otherwise be disjoint from the maximal independent set 
$\{a_{i}, b_{i}, \bar b_{i}\}$. Let $Q_{i} := \{\bar b_{i}, x_{i}\}$ if 
$x_{i} \in X$, and let $Q_{i} := \{b_{i}, \bar x_{i}\}$ otherwise. (Let us
recall that possibly $X$ contains neither of $x_{i}, \bar x_{i}$.)
Let
$$
X' := \{t_{1}, \dots, t_{q}\} \cup \{t'_{1}, \dots, t'_{q'}\} 
\cup \bigcup_{i=1}^{n} Q_{i}. 
$$
By definition $X'$ is in canonical form and $|X'|=k$. Thus it remains to show that
$X'$ is indeed a transversal. Using that $X$ is a transversal and that
\begin{itemize}
\item $\{t_{1}, \dots, t_{q}\} \cup \{t'_{1}, \dots, t'_{q'}\} \subseteq X'$; 
\item $X \cap \{x_{1}, \dots, x_{n}\} \subseteq X' \cap \{x_{1}, \dots, x_{n}\}$;
\item $X \cap \{\bar x_{1}, \dots, \bar x_{n}\} \subseteq X' \cap \{\bar x_{1}, \dots, \bar x_{n}\}$, and
\item $X \cap (\{y_{1}, \dots, y_{m}\} \cup \{\bar y_{1}, \dots, \bar y_{m}\}) = \varnothing$,
\end{itemize}
it directly follows that $X'$ meets each regular maximal independent set of $G$. 
Hence it remains to consider the non-regular ones. Clearly $X'$ intersects those 
of type 2, 3, 6, 7, 8, and 9. (Here we use that $q \geq 1$ and $q' \geq 1$.) As for types 4 and 5, 
the sets $S_{i}\cup T_{i}$ and $S'_{i}\cup T'_{i}$ have also a non-empty intersection with $X'$ for each 
$i\in [n]$, since $T_{i} \neq \varnothing$ and $T'_{i} \neq \varnothing$ by the properties of $\varphi$. 
Therefore $X'$ is a transversal, as desired. 
\end{proof}

So we now have a canonical transversal $X'$, and we wish to show that $\varphi$ holds.
We construct a truth assignment $\mu : \{x_1,x_2,\ldots ,x_n\}\mapsto \{\text{True}, \text{False}\}$ by letting $\mu (x_i)$ to True 
if and only if $x_i\not\in X'$. We need to show that this value for $x$ is a witness that $\varphi$ holds true. 

Consider a truth assignment $\mu'$ for $x,y$ that coincides with $\mu$ on $x$. We can assume that this assignment is not trivial in the sense that not every variable has the same 
truth value. Otherwise, since there exist at least one positive and one negative term, the formula is trivially satisfied. 
Furthermore, since the formula is nice, it is also trivially satisfied by the
assignment where only one negated literal $\bar{x}_i$ holds true, and all others variables are set to True, or, symmetrically where $x_i$ 
is set to True and all other variables to False. Hence we may assume that the considered assignment is not of this form either. 

Consider a maximal independent set $U$ containing exactly the literals set to True by $\mu'$.
With the above assumptions, the maximal independent set $U$ cannot be equal to $P$ or $N$, as they correspond to trivial truth assignment, where everything
is set to True (respectively, to False). Since $U$ contains one literal for each variable, it clearly cannot be of type \ref{stab:vargadgets}, \ref{stab:ptermgadgets}, or \ref{stab:ntermgadgets}. The set $U$ cannot be of the form $S_i\cup T_i, i\in [n]$ (type \ref{stab:pst}) either, as this case occurs when only one negated literal $\bar{x}_i$
holds true, and all others variables are set to True. Similarly, we can discard the case where $U$ is of the form $S'_i\cup T'_i, i\in [n]$ (type \ref{stab:nst}). 
Finally, $U$ cannot be of the form 
$W_{\ell} \cup Z_{\ell}$ for $\ell\in [q]$ (type \ref{stab:pwz}), nor of the form $W'_{\ell} \cup Z'_{\ell}$ for $\ell\in [q']$ (type \ref{stab:nwz}), since none of those two types of
independent sets contain one literal for each variable.

Hence the only remaining possibility is that $U$ is a regular maximal independent set. Then $U$ must contain a vertex $t_{\ell}$ for some $\ell\in [q]$, 
or a vertex $t'_{\ell}$ for some $\ell\in[q']$, as otherwise it would not be hit by $X'$. Suppose, without loss of generality, that $U$ contains $t_{\ell}$. 
Then $t_{\ell}$ is not adjacent to any literal in $U$, hence it must correspond to a term of $\varphi$ that is satisfied by the assignment. 
Since the reasoning holds for every assignment coinciding with $\mu$ on $x$, it shows that $\varphi$ holds true, as needed.
\end{proof}

As mentioned in the introduction, Theorem~\ref{thm:main} can equivalently be formulated in terms of clique 
transversals or fibres in posets. For completeness, we include these statements as corollaries. 

\begin{corollary}
Given a graph $G$ whose complement is bipartite and a positive integer $k$, it is 
$\sig$-complete to decide whether there is a clique transversal of $G$ of size at most $k$. 
\end{corollary}

\begin{corollary}
Given a poset $P$ and an integer $k$, deciding whether $P$ has a fibre of size at most $k$ is $\sig$-complete, 
even if $P$ has height two.
\end{corollary}

\section*{Acknowledgment}
We thank the anonymous referees for their careful reading of the manuscript and their helpful comments. 

\bibliographystyle{abbrv}
\bibliography{fiber}

\end{document}